\documentclass[11pt]{article}

\usepackage{amssymb}
\usepackage{amsmath}

\usepackage[T1]{fontenc}
\usepackage{concrete,eulervm}
\usepackage{parskip}

\usepackage{sectsty}
\allsectionsfont{\mdseries}

\usepackage{a4}
\usepackage{amsthm}
\usepackage{tikz}
\usepackage[round]{natbib}

\usepackage{hyperref}
\usepackage{xcolor}
\definecolor{dark-red}{rgb}{0.4,0.15,0.15}
\definecolor{dark-blue}{rgb}{0.15,0.15,0.4}
\definecolor{medium-blue}{rgb}{0,0,0.5}
\hypersetup{
    colorlinks, linkcolor={dark-red},
    citecolor={dark-blue}, urlcolor={medium-blue}
}

\theoremstyle{definition}
\newtheorem{theorem}{Theorem}
\newtheorem{remark}{Remark}
\newtheorem{lemma}{Lemma}

\newtheorem{proposition}{Proposition}

\newcommand{\mc}[1]{\mathcal{#1}}

\newcommand{\id}{\textnormal{id}}
\newcommand{\argmax}{\operatorname{argmax}}
\newcommand{\marg}{\textnormal{marg}}

\title{Interim Correlated Rationalizability in Large Games\thanks{We would like to thank Rabah Amir, Pierpaolo Battigalli, Martin Kaae Jensen, M. Ali Khan, Marcin Pęski and Xavier Vives for helpful discussions during the writing of this paper.}}
\author{\L{}ukasz Balbus\thanks{Universty of Zielona G{\'o}ra, Poland.},  Michael Greinecker\thanks{CEPS, ENS Paris-Saclay, France.}, Kevin Reffett\thanks{Arizona State University, USA.}, \L{}ukasz Wo\'{z}ny\thanks{SGH Warsaw School of Economics, Poland.}}

\begin{document}
\date{June 2025}
\maketitle

\begin{abstract}
We provide general theoretical foundations for modeling strategic uncertainty in large distributional Bayesian games with general type spaces, using a version of interim correlated rationalizability. We then focus on the case in which payoff functions are supermodular in actions, as is common in the literature on global games. This structure allows us to identify extremal interim correlated rationalizable solutions with extremal interim Bayes–Nash equilibria. Notably, no order structure on types is assumed. We illustrate our framework and results using the large versions of the electronic mail game and a global game.
\end{abstract}

\section{Introduction}
Much of economic uncertainty is strategic uncertainty—uncertainty about what agents are going to do. However, much of applied economic modeling relies on equilibrium notions that are ill-suited to the study of strategic uncertainty. In equilibrium, every agent reacts optimally to what other agents are doing, which usually requires knowing what those other agents are doing. To circumvent this limitation imposed by equilibrium concepts, many researchers have introduced incomplete information into game-theoretic models. Under incomplete information, although every agent knows how every other agent would react to each possible piece of information, not knowing the actual information makes them uncertain about how others are going to act.

The program of studying strategic uncertainty by introducing incomplete information into the model and then studying the resulting equilibria was particularly successful in the literature on global games, starting with the pioneering work of \citet{CarlssonVanDamme1993}. Global games are a specific way to model incomplete information by introducing a small amount of uncertainty into the payoffs of a complete information game with strategic complementarities. In particular, nature draws a state from a common prior, and agents receive a private signal given by a commonly known signal technology. This perturbation, while typically small, can drastically change the behavior of players and lead to a unique equilibrium, unlike the often multiple equilibria seen in complete information games. Though the literature on global games started with two-player games, models with a continuum of players proved to be particularly useful for economic application; see for example the early survey of \citet{morris_shin_2003}. Much of this literature has been quite open about the fact that the common prior and the prior signals serve mostly as a roundabout way to introduce strategic uncertainty and break common certainty of economic behavior. Moreover, though the global games literature superficially studies mostly Bayes-Nash equilibria of games of incomplete information, many of the arguments in the background work by iteratively eliminating dominated strategies, as is done, for example, by \citet{atkeson2000rethinking}. What drives the results in the global games literature seems to be a form of rationalizability. \citet{morris_shin_2003} go so far as to write that ``[t]he natural way to understand the `trick' to global games analysis is to go back and understand what is going on in terms of higher-order beliefs.'' But higher-order beliefs derived from a common prior have a very special structure, and this structure clouds the role higher-order beliefs play. \citet*{MorrisShinYildez2016} finally managed to analyze global games directly in terms of higher order beliefs, but their analysis is restricted to games with finitely many players.\footnote{\citet{MATHEVET2014252} provides bounds on the set of rationalizable solutions in supermodular games with finitely many players but requires type spaces to have canonical orders.} The authors restrict themselves to studying games with finitely many players not because the logic of their argument breaks down otherwise but because they have no appropriate model of higher-order beliefs with a continuum of agents. This paper supplies the needed tools.

We first show that one can meaningfully formulate large Bayesian games in an interim way based on type spaces. We formulate large games in distributional form, following \citet{MR777118}; they are not given in terms of an underlying set of players but are specified by a population distribution of characteristics. We then show that one can meaningfully study rationalizability in this setting. To do so, we adapt the notion of interim correlated rationalizability due to \citet*{dekel2007interim}. Interim correlated rationalizability is the solution concept corresponding to fixing a hierarchy of beliefs regarding the underlying state of nature, extending this hierarchy of beliefs to beliefs over action profiles, and imposing the requirement that everyone is rational and that everyone is commonly believed to be rational.\footnote{An explicit epistemic characterization along these lines is given by \citet*{battigalli2011interactive}. A different, more restrictive, notion of rationalizability for games of incomplete information has been proposed by \citet{ely2006hierarchies}. Their solution concept requires a conditional independence condition that is less natural in our distributional framework, in which the correlation of beliefs between players cannot be formalized.} We closely follow the formulation of \citet{RePEc:eee:mateco:v:72:y:2017:i:c:p:82-87} who also allow for a continuum of actions and states of nature, but restrict themselves to working with finitely many players. Our first major result, Theorem \ref{selfrat}, shows the equivalence between two formulations of interim correlated rationalizability. Our official definition is based on iterative elimination of actions that are not best replies to any belief consistent with nobody playing actions that have been previously eliminated. The second formulation then characterizes interim correlated rationalizability as the most permissive self-enforcing (set-valued) theory of rational behavior.

Most games studied in the literature on global games exhibit strategic complementarities. Such games have very nice properties. There exists a largest and a smallest equilibrium, and all rationalizable solutions are sandwiched in between; see \citet{milgrom1990rationalizability}. We show that a similar result holds for our large interim Bayesian games when we introduce strategic complementarities; this is our Theorem \ref{equilibrium}. As a byproduct, we obtain the first existence result for Bayes-Nash equilibria in large games with a continuum of agents and states of nature. An important question in the global games literature is equilibrium uniqueness. Indeed, the fundamental result of \citet{CarlssonVanDamme1993} shows that under certain conditions, there is a unique equilibrium in the limit as uncertainty vanishes. It follows from our result that there exists a unique interim Bayes-Nash equilibrium if and only if there exists a unique interim correlated rationalizable action for each player and each of their beliefs if and only if the largest and the smallest interim Bayes-Nash equilibrium coincide. We also provide sufficient epistemic conditions for equilibrium uniqueness.

Throughout, the type spaces that we use to model higher-order uncertainty about the state of nature are assumed to be compact metrizable. We show that this is essentially without loss of generality when the underlying states of nature form a compact metrizable space. In that case, we show (in Appendix 3) there exists a canonical type space inducing all (suitably coherent) hierarchies of beliefs. This result is of independent interest and opens the door to epistemic analyses of large games.

To illustrate the practicality of our approach, we provide two examples. First, we give a large game version of the coordinated attack problem and show how phenomena familiar from the electronic mail 
game of \citet{rubinstein1989electronic} can occur. The application shows how one can embed the usual equilibrium framework with private information coming from a common prior in our framework. Second, we show that the epistemic condition given by \citet*{MorrisShinYildez2016} for uniqueness in finite global games in terms of $p$-common certainty of uniform rank beliefs still works in large global games. In particular, our framework provides exactly the language to make the arguments work in large games.\bigskip

We want to point out how our work relates to some literature we have not mentioned so far. Rationalizability in large games of complete information has been studied by \citet{JaraMoroni2012}, \citet{Yu2014}, and \citet{greinecker2017} in individualistic models with an explicit set of players. Much of the work in these papers is dedicated to establishing topological closure properties that are automatic in our distributional framework. 

As mentioned above, rationalizable outcomes in supermodular games are sandwiched between extremal Bayes-Nash equilibria. We rely on methods that have been used to prove the existence of Bayes-Nash equilibria in supermodular games. In particular, we draw on the work of \citet{VanZandt2010JET}, who studies the existence of interim Bayes-Nash equilibria in supermodular games of incomplete information with finitely many players, and \citet*{MR3338692}, who study existence in a distributional model that is similar to ours but much more restrictive in its modeling of uncertainty.   

For applications, our toolkit makes it possible to study the rich applications of global games without the common prior assumption. \citet{IzmalkovYildiz} and \citet*{angeletos2018quantifying} construct specific parametric global games models without common priors; our tools make it possible to study higher-order uncertainty in full generality. Our mentioned adaptation of the uniqueness result of \citet*{MorrisShinYildez2016} to a continuum of players shows that this can be done on a practical level.

\section{Prelimininaries}

If $(X,\mc{X})$ is a measurable space, we let $\Delta(X)$ be the space of probability measures on $(X,\mc{X})$ endowed with the $\sigma$-algebra generated by functions of the form $\mu\mapsto \mu(E)$ with $E\in\mc{X}$. If $X$ is Polish (separable and completely metrizable) and $\mc{X}$ its Borel $\sigma$-algebra, our $\sigma$-algebra on $\Delta(X)$ is just the Borel $\sigma$-algebra induced by the topology of weak convergence of measures; see \citet[Proposition 7.25]{MR511544}. Generally, we endow the (at most) countable product of Polish spaces with the product topology, which is again Polish, and endow every Polish space with its Borel $\sigma$-algebra, and we endow products of measurable spaces with the product $\sigma$-algebra. These conventions are consistent. If $X$ and $Y$ are nonempty Polish spaces and $\nu$ is a Borel probability measure on $X$, we write $\Delta_\nu(X\times Y)$ for the space of Borel probability measures on $X\times Y$ with $X$-marginal $\nu$. If $Y$ is compact, then $\Delta_\nu(X\times Y)$ is compact too (in the topology of weak convergence). The proof is a simple application of Prohorov's ``tightness''-characterization of relative compactness.
\bigskip

\section{Large Bayesian games}

Our game theoretical model will specify both the fundamentals, given by player characteristics, actions, states of nature, and payoffs, and the beliefs and higher order beliefs of players regarding the state of nature.

The model is a distributional model in which individual players are not explicitly specified; only their population distribution is part of the model. Implicitly, we assume that individual players are insignificant; their behavior does not influence any aggregates. The distributional approach to large games was first introduced by \citet{MR777118}. Each player's fundamentals are specified by their characteristics. There is a Polish space $C$ of \emph{characteristics} whose population distribution is given by a Borel probability measure $\nu$ on $C$. All actions are included in a compact metrizable action space $A$. Not all actions need to be available to all players. There is an \emph{action correspondence} $\mc{A}:C\to 2^A$ that specifies for a player with characteristic $c$ the set $\mc{A}(c)$ of actions available to them. An \emph{action profile} is an element of $\Delta_\nu(C\times A)$ supported on the graph of $\mc{A}$. There is also a compact metrizable space $S$ of \emph{states of nature}. To tie everything together, payoffs are encoded in a single \emph{payoff function} $v:C\times A\times S\times\Delta_\nu(C\times A)\to\mathbb{R}$. In everything that follows, we assume that $\mc{A}$ is upper hemicontinuous with nonempty and compact values and that $v$ is continuous in every argument.

\begin{remark}
One could weaken the continuity assumption by merely assuming that $\mc{A}$ is measurable and by requiring $v$ only to be measurable in $C$. This would not be more general; there would still exist a Polish topology on $C$ under which $\mc{A}$ is upper hemicontinuous (continuous, actually) and for which $v$ would be continuous. Indeed, one can take under the weaker assumption $\mc{A}:C\to 2^A$ to be a measurable function for the Hausdorff topology on the space of nonempty compact subsets of $A$ by \citet[Theorem 18.10]{MR2378491} and identify $v$ with a measurable function $u:C\to C\big[A\times S\times\Delta_\nu(C\times A)\big]$ given by $u(c)(a,s,\mu)=v(c,a,s,\mu)$ by \citet[Theorem 4.55]{MR2378491}. Now by \citet[Theorem 4.59]{MR2378491}, there exists a finer Polish topology on $C$ that introduces no new Borel sets such that the two measurable functions $\mc{A}$ and $u$ are continuous. Under this finer topology, $\mc{A}$ is a continuous correspondence by \citet[Theorem 17.15]{MR2378491}. It is also easy to see that $v$ is continuous in all arguments if $u$ is continuous.
\end{remark}

To complete the specification of the Bayesian game, we have to specify beliefs and higher-order beliefs. We do this at the interim stage in terms of type spaces. For technical convenience, we follow \citet{MR558626} and \citet{MR784702} and use spaces of states of the world that have no inherent product structure. A \emph{type space} is a triple $(T,\sigma,\tau)$ with $T$ a nonempty compact metrizable space of \emph{states of the world} and two continuous functions $\sigma:T\to S$ and $\tau:T\to\Delta_\nu\big(C\times\Delta(T)\big)$. The idea is that a state of the world specifies everything relevant in the model. In particular, it specifies the state of nature via the function $\sigma$. It also specifies the joint distribution of characteristics and beliefs. Since beliefs should be specified on everything relevant, they are specified for states of the world. Note that $\tau$ has values in $\Delta_\nu(C\times\Delta(T))$ and not in $\Delta_\nu(C\times T)$. From an element of $\Delta_\nu\big(C\times\Delta(T)\big)$ we obtain a regular conditional probability as a function from $C$ to $\Delta\big(\Delta(T)\big)$. So for each characteristic, there is a distribution over beliefs. In particular, beliefs need not be a function of characteristics; we do not formally include beliefs as part of the characteristics. Note also that we interpret players as individually insignificant, so we need no introspection condition for type spaces; no player has explicit beliefs over their own beliefs. 

The assumption that $T$ is compact metrizable (and $\sigma$ continuous) is essentially without loss of generality. Interim correlated rationalizability, as defined in the next section, is a monotone solution concept in the sense that having more hierarchies of beliefs allows for more interim correlated rationalizable solutions. One can show that a compact metrizable type space exists that includes all hierarchies of beliefs that satisfy minimal coherency conditions, provided that $S$ is compact metrizable. Indeed, each state of the world $t$ induces are unique state of nature $\sigma(s)$. So, a belief about the state of the world induces a first-order belief about $S$. A state of nature $t$ also induces a population distribution $\tau(t)$ of beliefs about the state of the world and, thus, beliefs about the state of nature. Since we care about how beliefs and states of nature relate, a belief about states of the world induces a second-order joint belief about the state of nature and about the population distribution of first-order beliefs. Since second-order beliefs include first-order beliefs as their $S$-marginals, these derived first-order beliefs coincide with the actual first-order beliefs. A similar consistency condition applies to all higher-order beliefs, too; this is the condition that hierarchies of beliefs, the whole sequence of all higher-order beliefs, are \emph{coherent}. Moreover, the hierarchies of beliefs induced by a type space can be identified with joint beliefs over the state of nature and coherent hierarchies of beliefs that, in turn, only consider coherent hierarchies of beliefs. And so on. We show in Appendix 3 that a canonical type space exists that includes all hierarchies of beliefs satisfying the restrictions imposed by such coherency conditions and, therefore, all hierarchies of beliefs occurring in any type space. If $S$ is compact metrizable, so is the canonical type space.

\section{Interim correlated rationalizability}

Our first solution concept is interim correlated rationalizability. As we will see, there are two equivalent ways to define interim correlated rationalizability. Our official definition is based on iterative elemination of actions that are not best replies to beliefs that put no mass on any previously eliminated actions. The important property of interim correlated rationalizability is that otherwise arbitrary correlation between action choices and states of nature are allowed. The restrictions on actions are given by a correspondence $\mc{S}:C\times\Delta(T)\to 2^A$. We can think of $\mc{S}$ as a theory of behavior that specifies for each characteristic of a player and for each belief of that player what actions they could play. Provided $\mc{S}$ has a measurable graph, we get a well-defined correspondence $\mc{D}_\mc{S}:T\to 2^{\Delta_\nu\big(C\times\Delta(T)\times A\big)}$ by specifying
\[\begin{split}
\mathcal{D}_\mc{S}(t)=\Big\{\kappa\in\,\Delta_\nu\big(C\times \Delta(T)\times A\big)\mid &\kappa\textnormal{ is supported on the graph of } \mc{S}\\ & \textnormal{and has }C\times\Delta(T)\textnormal{-marginal }\tau(t)\Big\}.
\end{split}\]
We can interpret $\mathcal{D}_\mc{S}(t)$ as the space of all belief-action profiles feasible at the state of the world $t$ if everyone behaves according to $\mc{S}$. 

\begin{lemma}\label{Duhc}If $\mc{S}$ is an upper hemicontinuous correspondence with nonempty compact values, then so is $\mc{D}_\mc{S}$.
\end{lemma}

Let $\phi:\Delta_\nu\big(C\times\Delta(T)\times A\big)\to \Delta_\nu(C\times A)$ be the canonical marginal-function. If $\mc{D}:T\to 2^{\Delta_\nu\big(C\times\Delta(T)\times A\big)}$ has a measurable graph, we get a well-defined correspondence $\mc{S}_\mc{D}:C\times\Delta(T)\to 2^A$ by specifying
\[\begin{split}
\mc{S}_{\mc{D}}(c,\beta)=\bigg\{a\in A~\Bigl\lvert&~\textnormal{ there is some }\mu\in\Delta\big(T\times\Delta_\nu(C\times \Delta(T)\times A)\big)\textnormal{ with }\\ &T\textnormal{-marginal }\beta\textnormal{ and supported on the graph of }\mc{D}\\ &\textnormal{such that }a\in\argmax_{\mc{A}(c)} \int v\big(c,a,\sigma(t),\phi(\kappa)\big)~\mathrm d\mu(t,\kappa)
\bigg\}.
\end{split}\]
Intuitively, $\mc{S}_\mc{D}(c,\beta)$ is the set of actions a rational player with characteristic $c$ and belief $\beta$ can play when being certain that everyone behaves according to $\mc{D}$.
\begin{lemma}\label{Sdef} If $\mc{D}$ is an upper hemicontinuous correspondence with nonempty compact values, then so is $\mc{S}_{\mc{D}}$.
\end{lemma}

We now define recursively two sequences of correspondences $\langle \mc{S}_m\rangle$ with $\mc{S}_m:C\times\Delta(T)\to 2^A$ and $\langle\mc{D}_m\rangle$ with $\mc{D}_m:T\to 2^{\Delta_\nu\big(C\times\Delta(T)\times A\big)}$. We let $S_0(c,\beta)=\mathcal{A}(c)$ for all $(c,\beta)$ and let $\mathcal{D}_m=\mc{D}_{\mc{S}_m}$. Given that $\mathcal{D}_m$ is defined, we define $\mc{S}_{m+1}$ by $\mc{S}_{m+1}=\mc{S}_{\mc{D}_m}$. It follows from Lemma \ref{Duhc} and Lemma \ref{Sdef} that $\mc{S}_m$ and $\mc{D}_m$ are upper hemicontinuous correspondence with nonempty compact values. We let $\mc{S}:C\times\Delta(T)\to 2^A$ be defined by $\mc{S}(c,\beta)=\bigcap_m\mc{S}_m(c,\beta)$. It follows readily that $\mc{S}$ is also an upper hemicontinuous correspondence with nonempty compact values from \citet[Theorem 17.25]{MR2378491}. The correspondence $\mc{S}$ represents our solution concept of \emph{interim correlated rationalizability}. The action $a$ is \emph{interim correlated rationalizable} for a player with characteristics $c$ and belief $\beta$ if $a\in\mc{S}(c,\beta)$.

We now provide an alternative characterization of interim correlated rationalizability. An upper hemicontinuous nonempty and compact valued correspondence $\mathcal{R}:C\times\Delta(T)\to 2^A$ \emph{rationalizes itself} if $\mc{R}=\mc{S}_{\mc{D}_\mc{R}}$. Intuitively, a theory of behavior rationalizes itself if everyone who believes that everyone behaves according to the theory behaves rationally by following the theory.

\begin{theorem}\label{selfrat}$\mathcal{S}$ is the largest correspondence that rationalizes itself.\end{theorem}

Note that interim correlated rationalizability is a monotone solution concept in that larger type spaces admit more interim correlated rationalizable solutions. Consequently, it is without loss of generality to work with the canonical type space that includes all combinations of a state of nature and a population distributions of hierarchies of beliefs from any type space. Further compactness and continuity requirements on types spaces we have imposed are, therefore, innocent.

\section{Strategic complementarities and the interim Bayes-Nash equilibrium}

We now impose more structure on the Bayesian games we study. We let $\succeq$ be a lattice ordering on $A$ under which the join and meet are continuous operations. It follows that $(A,\succeq)$ is a complete lattice with a closed graph; see \citet{Reny2011ECTA} for the details. We use $\succeq$ also to denote the stochastic order on $\Delta_\nu(C\times A)$ and other closely related orders derived from it. This abuse of notation should not lead to confusion. (Background material on stochastic orderings is provided in Appendix 1). 

We make the following additional assumptions:

\begin{itemize}
\item[(i)] $\mathcal{A}(c)$ is a sublattice of $A$ for each $c\in C$ (closed under the original join and meet operators).
\item[(ii)] The function $v$ is supermodular in $A$, so for each $a,a'\in A$, $c\in C$, $s\in S$, and $\mu\in\Delta_\nu(C\times A)$
\[v(c,a\vee a',s,\mu)+v(c,a\wedge a',s,\mu)\geq v(c,a,s,\mu)+v(c,a',s,\mu).\]
\item[(iii)] The function $v$ has increasing differences in $A\times\Delta_\nu(C\times A)$, so for each $a\succeq a'$, the function
\[\mu\mapsto v(c,a,s,\mu)-v(c,a',s,\mu)\]
is $\succeq$-nondecreasing.
\end{itemize}
Note that $\Delta_\nu(C\times A)$ need not be a lattice even when $A$ is; see \citet*[page 901]{MR0494447} for an example.

Under the new assumptions, we can show that there exists a largest and a smallest interim Bayes-Nash equilibrium and that every agent's interim correlated rationalizable actions are sandwiched in between.

For a correspondence $\phi$ with values in $A$, we let $\bigvee\phi$ and $\bigwedge\phi$ be the functions with values in $A$ given by pointwise suprema and infima, respectively.

A strategy profile specifies for each characteristic of an agent and each of their beliefs the conditional distribution over actions. A symmetric strategy profile will be a strategy profile in which everyone with the same type and same belief will choose the same action. So a \emph{symmetric strategy profile} is a measurable function $\zeta:C\times\Delta(T)\to A$ such that $\zeta(c,\beta)\in\mc{A}(c)$ for all $c$. In the following, we will always mean symmetric strategy profiles if we talk about a strategy profile.

If $\zeta$ is a strategy profile, we let $\zeta^*:\Delta_\nu\big(C\times\Delta(T)\big)\to\Delta_\nu\big(C\times\Delta(T)\times A\big)$ be the induced function given by
\[\zeta^*(\kappa)(B)=\int 1_B\big(c,\beta,\zeta(c,\beta)\big)~\mathrm  d\kappa(c,\beta)\]
for each Borel set $B \subseteq C\times\Delta(T)\times A$.
%

The strategy profile $\zeta$ is an \emph{(interim) Bayes-Nash equilibrium} if for all $c\in C$ and $\beta\in\Delta(T)$, $\zeta(c,\beta)$ is an element of
\[\argmax_{\mc{A}(c)} \int v(c,\cdot,s,\mu)~\mathrm d\beta\circ(\sigma,\phi\circ\zeta^*\circ\tau)^{-1}.\]

\begin{theorem}\label{equilibrium}The functions $\bigvee\mc{S}$ and $\bigwedge\mc{S}$ are both Bayes-Nash equilibria.
\end{theorem}

We defined Bayes-Nash equilibrium only for symmetric strategy profiles for the sake of convenience. Weakening the symmetry requirement would not change the result, but involve additional notational complications. It should be noted that the proof of Theorem \ref{equilibrium} relies on methods from \citet{VanZandt2010JET} and \citet*{MR3338692}. 

\begin{remark}We work with large games in distributional form, but Theorem \ref{equilibrium} also implies an existence theorem for large games with an explicit probability space of individual players. This is due to our equilibria being symmetric. Each individualistic large game induces a canonical distributional game. Taking the composition of the function from types to characteristics and a symmetric equilibrium provides us with a pure strategy equilibrium for the individualistic game.
\end{remark}

\section{The canonical type space}

Throughout, we have assumed that $T$ is compactly metrizable. Interim correlated rationalizability is a monotone solution concept in the sense that having more hierarchies of beliefs allows for more interim correlated rationalizable solutions. We now show that a compact metrizable type space exists that includes all hierarchies of beliefs that satisfy minimal coherency conditions, provided that $S$ is compact metrizable. 

Let $(T,\sigma,\tau)$ be a type space. Each state of the world $t$ induces are unique state of nature $\sigma(s)$. So, a belief about the state of the world induces a first-order belief about $S$. A state of the world $t$ also induces a population distribution $\tau(t)$ of beliefs about the state of the world and, thus, beliefs about the state of nature. Since we care about how beliefs and states of nature relate, a belief about states of the world induces a second-order joint belief about the state of nature and about the population distribution of first-order beliefs. Since second-order beliefs include first-order beliefs as their $S$-marginals, these derived first-order beliefs coincide with the actual first-order beliefs. A similar consistency condition applies to all higher-order beliefs too; this is the condition that hierarchies of beliefs, the whole sequence of all higher-order beliefs, are \emph{coherent}. Moreover, the hierarchies of beliefs induced by a type space can be identified with joint beliefs over the state of nature and coherent hierarchies of beliefs that, in turn, only consider coherent hierarchies of beliefs. And so on. We show in Appendix 3 that a canonical type space exists that includes all hierarchies of beliefs satisfying the restrictions imposed by such coherency conditions and, therefore, all hierarchies of beliefs occurring in any type space. If $S$ is compact metrizable, so is the canonical type space.

\section{Example: Coordinated attack 
}
There is a continuum of players (normalized to 1) that need to coordinate to attack one of two spots (risky or safe 
).
There are hence two actions, where $a=1$ denotes a decision to attack the risky and $a=0$ the safe spot.
There are two states of nature determining the game that will be played. If $s=0$ the payoffs are as follows:
$$
\left\{
\begin{array}{lll}
M(1-\vartheta)&\mbox{if}&a=0\\
-L(1-\vartheta)&\mbox{if}&a=1,\\
\end{array}
\right.
$$
where $\vartheta$ denotes the fraction of players choosing $a=1$ and $M$ and $L$ are parameters that satisfy $L>M>0$. 
If $s=1$ the payoffs are as follows:
$$
\left\{
\begin{array}{lll}
0&\mbox{if}&a=0\\
M\vartheta-L(1-\vartheta)&\mbox{if}&a=1.\\
\end{array}
\right.
$$
For each state, players play a coordination game. Indeed, the marginal payoff from taking action $1$ (versus $0$) increases with $\vartheta$ for each state.
Each game has two symmetric Nash equilibria: the greatest one, where all players attack the risky spot ($\vartheta^*=1$) and the least one, in which none of players attacks the risky spot ($\vartheta^*=0$). Players do not observe the state.
There is a common prior with the probability that $s=1$ given by $\pi$. For each $\pi\in[0,1]$, again there are two Nash equilibria: $\vartheta^*=1$ and $\vartheta^*=0$ of the incomplete information game.


Every player has a position on the circle of unit circumference $[0,1)$. Players at position $0$ are perfectly informed of the true state and receive a positive signal if the state is $1$. The signal travels along the circle, potentially passing the same position several times, until it dies at some random time. The mechanism is analogous to the one used by 
\citet{rubinstein1989electronic}. When we refer to the number of signals a player receives, we mean the number the signal passed them by.

Formally, we let $C=[0,1)$ and $\nu$ be the uniform distribution. The only relevant characteristic of a player is their position. We take $S=A=\{0,1\}$. The payoff function $v$ is given by
\[v(c,a,s,\mu)=
\left\{
\begin{array}{lll}
M\mu_A(\{0\})&\mbox{if}&s=0,a=0\\
-L\mu_A(\{0\})&\mbox{if}&s=0,a=1\\
0&\mbox{if}&s=1,a=0\\
M\mu_A(\{1\})-L\mu_A(\{0\})&\mbox{if}&s=1,a=1,\\
\end{array}
\right.\]
with $\mu_A$ denoting the $A$-marginal of $\mu\in\Delta_\nu(C\times A)$.

We let $T=\{0,0\}~\cup~\{1\}\times[0,\infty]$. The second coordinate denotes the time when the signal dies. If this time is $\infty$, the signal never dies, a situation corresponding to common certainty of the state of nature. The function $\sigma$ is simply the projection onto $S$. The function $\tau$ will be defined below.

We assume the time the signal dies follows an exponential distribution (analogous to the geometric distribution used by Rubinstein) with intensity parameter $\alpha>0$. The corresponding cumulative distribution function is given for nonnegative $x$ by
\[F(x)=1-e^{-\alpha\,x}.\]
The expected time the signal dies is $1/\alpha$.

Next, we specify $\tau$. A player's belief will be fully determined by the player's position and the time the signal dies. If a player positioned at $i$ receives no signal, then either the state of nature is $0$ or the state of nature is $1$, but the signal died before reaching $i$. The conditional probability of latter happening is given by \[\pi_{i}=\frac{\pi(1-e^{-\alpha i})}{1-\pi+\pi(1-e^{-\alpha i})}.\]
The belief of a player who received no signal that the signal died before $x$ conditional on the state of nature being $1$ is given by \[\frac{1-e^{-\alpha x}}{1-e^{-\alpha i}}\quad \mbox{for }x\in (0,i).\]
If a player receives at least one signal, then they are certain that the state of nature is $1$. If a player positioned at $i$ receives exactly $n>0$ signals, the signal must have died in the interval $(ni,ni+1)$. The conditional belief of such a player that the signal died after $ni$ but before $ni+x$ is given by 
\[\frac{1-e^{-\alpha x}}{1-e^{-\alpha}}\quad \mbox{for }x\in (0,1).\]
These conditional probabilities are enough to specify the function $\tau$.

The following proposition mirrors the central result of \cite{rubinstein1989electronic}.
\begin{proposition}
The equilibrium in which all players always play $0$ is the only equilibrium in which everyone who did not observe a signal plays $0$. 
\end{proposition}
\begin{proof}Suppose there were an equilibrium in which everyone who did not observe a signal plays $0$ but in which action $1$ is sometimes played. We say $t$ is \emph{active} if the player $t\mod 1$ plays $1$ after observing $\lfloor t \rfloor+1$ signals.
Let $t^*$ be the essential infimum of all active times. By assumption, $t^*<\infty$. Let $i^*=t^*\mod 1$. If $i^*$ receives $\lfloor t^* \rfloor+1$ signals, they know that the only players that might play $1$ are those that that have a larger position and received the same number of signals or have a lower position and received one more signal. The expected fraction of such players is
\[\int_{0}^1x \bigg(\frac{1-e^{-\alpha x}}{1-e^{-\alpha}}\bigg)' dx=\int_{0}^1x\frac{\alpha e^{-\alpha x}}{1-e^{-\alpha}}dx=\frac{1}{\alpha}-\frac{1}{e^{\alpha}-1}.\]
Taking action $1$ can only be optimal if 
$M\vartheta-(1-L)\vartheta\geq 0$ or, equivalently
\[\vartheta\geq \frac{L}{M+L}.\]
We must have, therefore, 
\[\frac{1}{\alpha}-\frac{1}{e^{\alpha}-1}\geq\frac{L}{M+L}>\frac{1}{2}.\] But the function
\[\alpha\mapsto \frac{1}{\alpha}-\frac{1}{e^{\alpha}-1}\] is strictly decreasing on the strictly positive real line, and has limit $1/2$ as $\alpha$ decreases to $0$. So, this inequality can never be satisfied. 
\end{proof}

\section{Example: Equilibrium uniqueness in global games}

We show here that the arguments of \citet*{MorrisShinYildez2016} for uniqueness in finite global games in terms of higher order beliefs can be formulated in our setting; our framework is flexible enough. 

In our example, characteristics play no role and are notationally suppressed. The action space is binary, $A=\{0,1\}$ (``noninvest'' and ``invest'') and the states of nature is  given by $S=[-1,2]$. We write $\vartheta$ for the fraction of players playing $1$. Payoff functions are given, by slight abuse of notation, by
\[v(a,s,\vartheta)=a\cdot (s+\vartheta-1).\]
Under this specification, the game is supermodular. In particular, there is a largest and a smallest equilibrium, given by  $\bigvee \mathcal{S}:\Delta(T)\to A$ and $\bigwedge \mathcal{S}:\Delta(T)\to A$, respectively, and every rationalizable solution lies in between by Theorem \ref{equilibrium}. Since there are no proper characteristics, the extremal equilibria specify what an agent with a given belief might do. In particular, we are able to discuss when a belief gives rise to a unique rationalizable action, even in settings in which not everyone has a unique rationalizable action. This distinction is somewhat buried when there are only two players but is clear in our setting.

Our game is linear, and we make heavy use of the fact that usually only averages matter. For example, we let $x_\beta=x(\beta)=\int \sigma~\mathrm d\beta$ be the expected state under the belief $\beta$. Also, we let the expected fraction of players with a belief in $E\subseteq\Delta(T)$ under the belief $\beta$ be given by $\text{F}_\beta(E)=\int \tau(E)~\mathrm d\beta$. For each player with belief $\beta$, the expected fraction of players that are less optimistic about the state of nature than them is their \emph{rank}, formally defined as
\[\text{R}(\beta)=\text{F}_\beta\big(\big\{\beta'\in\Delta(T)\mid x_{\beta'}\leq x_\beta\big\}\big).\]
An important role in our uniqueness result will be played by the set beliefs of players who believe their rank to be close to the median rank. For $\epsilon>0$, let
\[\text{URB}_\epsilon=\big\{\beta\in\Delta(T)\mid 1/2-\epsilon\leq \text{R}(\beta)\leq 1/2+\epsilon\big\}.\] Another important set is the set of beliefs for which investing is ``$\epsilon$-strictly risk dominant,'' given by
\[\text{SRD}_\epsilon=\big\{\beta\in\Delta(T)\mid x_\beta>1/2+\epsilon\big\}.\]
Dually,
\[\text{nSRD}_\epsilon=\big\{\beta\in\Delta(T)\mid x_\beta<1/2-\epsilon\big\}.\]
is the set of beliefs for which noninvesting is strictly risk dominant.
Our last ingredient is suitable (approximate) belief operators and certainty operators. For reasons we will explain below, we define these operators on $\Delta(T)$. For a measurable function $f:\Delta(T)\to\mathbb{R}$, define a corresponding belief operator by
\[\text{B}_f(E)=\big\{\beta\in E\mid \text{F}_\beta(E)\geq f(\beta)\big\}.\] In particular, for $f=1-x$, we have
\[\text{B}_{1-x}(E)=\big\{\beta\in E\mid \text{F}_\beta(E)\geq 1-x(\beta)\big\},\]
and for $f=x$, we have
\[\text{B}_x(E)=\big\{\beta\in E\mid \text{F}_\beta(E)\geq x(\beta)\big\}.\]
We identify a number $p$ with the corresponding constant function $\Delta(T)$ when writing  $\text{B}_p$. For each operator $\text{B}_f$, we recursively define its iterates by $\text{B}^1_f(E)=\text{B}_f(E)$ and $\text{B}_f^{n+1}=\text{B}_f\big(B_f^n(E)\big)$, and the corresponding certainty operator by $\text{C}_f(E)=\bigcap_n \text{B}_f^n(E)$. We can use the operators $\text{C}_{1-x}$ and $\text{C}_{x}$ to characterize which actions could be played under a given belief.
\begin{remark}
 These belief and certainty operators operate on Borel subsets of $\Delta(T)$ and not on Borel subsets of $T$, which might seem more natural. Belief operators usually satisfy an introspection condition. In our framework, there is no difference between believing that everyone believes something and believing that everyone else believes something. However, we need the former to hold, and that is why we define the belief and certainty operators the way we do. In general, one could define belief operators that operate on both beliefs and sets of states of the world simultaneously to achieve greater generality. In this application, there is no need for such generality.  
\end{remark}

\begin{lemma} Investment can be optimal under $\beta$, $\bigvee \mathcal{S}(\beta)=1$, if and only if $\beta\in\text{C}_{1-x}\big(\Delta(T)\big)$. Similarly, noninvestment can be optimal under $\beta$, $\bigwedge \mathcal{S}(\beta)=0$, if and only if $\beta\in\text{C}_{x}\big(\Delta(T)\big)$.
\end{lemma}
\begin{proof}We do the first case; the second case is analogous. Indeed, $\beta\in \text{B}_{1-x}\big(\Delta(T)\big)$ is equivalent to $x_\beta\geq 0$, so investment is optimal if everyone else invests too. If everyone with beliefs in $\text{B}_{1-x}^n\big(\Delta(T)\big)$ invests, then
\[\text{B}_{1-x}^{n+1}\big(\Delta(T)\big)=\Big\{\text{F}_\beta(\text{B}_{1-x}^n\big(\Delta(T)\big))\geq 1-x(\beta)\Big\}\]
contains the beliefs under which investment is optimal.
\end{proof}
Importantly, since every player must either invest or not, investment is uniquely rationalizable if $\beta\notin\text{C}_{x}\big(\Delta(T)\big)$, and nonivestment is uniquely rationalizable if $\beta\notin \text{B}_{1-x}\big(\Delta(T)\big)$. 


For a measurable set $E\subseteq\Delta(T)$ and $\epsilon\geq 0$, let \[x^{**}_\epsilon(E)=\sup\{x_\beta\mid\beta\in E, x_\beta\leq \text{R}(\beta)+\epsilon\}\] and \[x^{\epsilon}_{**}(E)=\inf\{x_\beta\mid\beta\in E, x_\beta\geq \text{R}(\beta)-\epsilon\}.\] Here, $x^{**}_{\epsilon}(E)$ is the highest value (of the state) for which any type within $E$ has a value exceeding its rank by at most $\epsilon$. Similarly, $x^\epsilon_{**}(E)$ is the lowest value (of the state) for which any type within $E$ has a rank exceeding the corresponding value by at most $\epsilon$.

\begin{lemma}\label{investstars} Let $E\subseteq \Delta(T)$ be closed and $E\subseteq B_p(E)$. Investment is uniquely rationalizable for any $\beta\in E$ such that
\[x_\beta>x^{**}_{1-p}(E).\]
Similarly, noninvestment is uniquely rationalizable for any $\beta\in E$ such that
\[x_\beta<x_{**}^{1-p}(E).\]

\end{lemma}
\begin{proof}Again, we do the first case. By the last lemma, we have to show that the condition implies $\beta\notin\text{C}_{x}\big(\Delta(T)\big)$. This is trivially the case if $\text{C}_{x}\big(\Delta(T)\big)\cap E=\emptyset$, so we can assume that $\text{C}_{x}\big(\Delta(T)\big)\cap E\neq\emptyset$. We show that $\big\{x_\beta\mid \beta\in \text{C}_{x}\big(\Delta(T)\big)\cap E\big\}$ has a maximum. Since the function $\beta\mapsto x_\beta$ is continuous and $\Delta(T)$ compact, it suffices to show that $\text{C}_x(H)$ is closed if $H$ is, which reduces to showing that $\text{B}_x(H)$ is closed if $H$ is. This, in turn, reduces to showing that $\beta\mapsto \text{F}_\beta(H)$ is upper-semicontinuous if $H$ is closed, which follows from the Portmanteau theorem.

So let $\hat{\beta}$ be a maximizer of $x$ on $\text{C}_{x}\big(\Delta(T)\big)\cap E$ and $\hat{x}=x_{\hat{\beta}}$ the corresponding value. By construction, if $x_\beta>\hat{x}$, then $\beta\notin \text{C}_{x}\big(\Delta(T)\big)\cap E$. If we know that $\beta\in E$, this means $\beta\notin \text{C}_{x}\big(\Delta(T)\big)$ and investment is uniquely rationalizable at $\beta$. It suffices, therefore, to show that $x^{**}_{1-p}(E)\geq\hat{x}$.

Moreover,
\[\text{F}_{\hat{\beta}}\Big(\text{C}_x\big(\Delta(T)\big)\Big)=\text{F}_{\hat{\beta}}\Big(\text{C}_x\big(\Delta(T)\big)\cap E\Big)+\text{F}_{\hat{\beta}}\Big(\text{C}_x\big(\Delta(T)\big)\setminus E\big)\Big)\]
\[\leq\text{F}_{\hat{\beta}}\Big(\text{C}_x\big(\Delta(T)\big)\cap E\Big)+\text{F}_{\hat{\beta}}\big(\Delta(T)\setminus E\big)\]
\[=\text{F}_{\hat{\beta}}\Big(\text{C}_x\big(\Delta(T)\big)\cap E\Big)+\big(1-\text{F}_{\hat{\beta}}(E)\big)\]
\[\leq \text{F}_{\hat{\beta}}\Big(\text{C}_x\big(\Delta(T)\big)\cap E\Big)+(1-p),\]
where the last inequality uses that $\hat{\beta}\in E$ and $E\subseteq B_p(E)$, which implies $\text{F}_{\hat{\beta}}(E)\geq p$. Now, using the fact that $\hat{\beta}$ is a maximizer, we get
\[\text{F}_{\hat{\beta}}\Big(\text{C}_x\big(\Delta(T)\big)\cap E\Big)+(1-p)\leq\text{F}_{\hat{\beta}}\Big(\big\{\beta\in\Delta(T)\mid x_\beta\leq \hat{x}\big\}\Big)+(1-p)\]
\[=\text{R}(\hat{\beta})+(1-p).\]

Note that $\text{B}_x\big(\text{C}_x\big(\Delta(T)\big)\big)=\text{C}_x\big(\Delta(T)\big)$. Since $\hat{\beta}\in \text{C}_x\big(\Delta(T)\big)$ and by the above inequalities:
\[\hat{x}=x_{\hat{\beta}}\leq \text{R}(\hat{\beta})+(1-p).\]
We hence have:
\[\hat{x}\in\big\{x_\beta\mid\beta\in \Delta(X), x_\beta\leq \text{R}(\beta)+(1-p)\big\},\]

which implies
\[\hat{x}\leq\sup\big\{x_\beta\mid\beta\in E, x_\beta\leq \text{R}(\beta)+(1-p)\big\}=x^{**}_{1-p}(E).\]

\end{proof}

\begin{proposition}
If $C_{1-\epsilon}(URB_\epsilon)$ is closed, then investing is uniquely rationalizable for all 
\[\beta\in SRD_{2\epsilon}\cap C_{1-\epsilon}(URB_\epsilon).\]

\noindent If $C_{1-\epsilon}(URB_\epsilon)$ is closed, then not investing is uniquely rationalizable for all 
\[\beta\in nSRD_{2\epsilon}\cap C_{1-\epsilon}(URB_\epsilon).\]
\end{proposition}
\begin{proof}
Since, $C_{1-\epsilon}(URB_\epsilon)= B_{1-\epsilon}(C_{1-\epsilon}(URB_\epsilon))$, by Lemma \ref{investstars}, it suffices to show that $x_\beta>x^{**}_\epsilon\big(C_{1-\epsilon}(URB_\epsilon)\big)$ for all $\beta$ in the relevant intersection. 

By definition, $C_{1-\epsilon}(URB_\epsilon)\subseteq URB_\epsilon$. Therefore,
\[ x^{**}_\epsilon(C_{1-\epsilon}(URB_\epsilon))\leq x^{**}_\epsilon(URB_\epsilon)=\sup\{x_\beta\mid\beta\in URB_\epsilon, x_\beta\leq \text{R}(\beta)+\epsilon\}\leq 1/2+2\epsilon.\]
Now,
\[\text{SRD}_{2\epsilon}=\big\{\beta\in\Delta(T)\mid x_\beta>1/2+2\epsilon\big\},\]
so
\[\beta\in SRD_{2\epsilon}\cap C_{1-\epsilon}(URB_\epsilon)\]
implies $x_\beta>x^{**}_\epsilon\big(C_{1-\epsilon}(URB_\epsilon)\big)$ and we are done.
\end{proof}

\begin{remark} As mentioned above, we have defined belief and certainty operators on $\Delta(T)$ instead of $T$. In our large game model, there is no difference between a player's belief about what all players believe and what all other players believe. In particular, the model doesn't directly express introspective beliefs; beliefs about a player's own beliefs. But it is exactly such introspective beliefs that occur in the relevant definition. For example, rank beliefs are about how a player ranks the optimism of their own beliefs relative to the rest of the population.
\end{remark}

\section{Discussion}

In this paper, we have formulated large games of incomplete information and a suitable notion of interim correlated rationalizability, showed that the model is not overly restrictive by constructing the canonical type space, and showed that under strategic complementarities, all rationalizable outcomes are bracketed by extremal equilibria. As a byproduct, we have an existence result for pure-strategy equilibria in large games with nontrivial aggregate uncertainty and strategic complementarities. The model is flexible enough to subsume traditional approaches with a common prior, as we do in our coordinated attack problem, and to adapt interim arguments in terms of beliefs, as we do in our adaption of the uniqueness result of \citet*{MorrisShinYildez2016}. In this section, we want to discuss some technical aspects of our approach as well as some wider applications.

In our formulation, all players are implicitly negligible and small. There are applications in which it is natural to have large players such as governments and central banks, alongside the small players. Including such large players in the model would require a lot more notation, but would not give rise to any technical problems. We can embed a large player as an isolated point in the space of characteristics $C$. If $c\in C$ represents a large player, an element $\mu\in\Delta_\nu(C\times A)$ that represents a feasible joint distribution over types and actions has to satisfy $\nu(c)=\mu(c, a)$ for some $a\in A$ (abusing notation, we identify singletons and their elements here). To show that such conditions give rise to a closed subset, one can adapt the proof of \citet[Proposition 7.2]{LEVY2024103012}. For the construction of the canonical type space, one would need to impose a familiar introspection condition for large players.

We have assumed that $A$ is compact to ensure that $\Delta_\nu(C\times A)$ is compact. However, we only need that the elements of $\Delta_\nu(C\times A)$ supported on the graph of $\mathcal{A}$ are compact. If $A$ is merely Polish but $\mathcal{A}$ compact-valued, this holds by Prohorov's characterization of relative compactness. If $K\subseteq C$ is a compact set such that $\nu(K)>1-\epsilon$, then $K\times\mathcal{A}(K)$ is compact, because $\mathcal{A}$ is upper hemicontinuous and compact-valued, and $\mu(K\times\mathcal{A}(K))>1-\epsilon$ for all $\mu\in\Delta_\nu(C\times A)$ supported on the graph of $\mathcal{A}$

A more substantive restriction of our model is that it is entirely static. Global games have been extended to dynamic models such as the dynamic regime change model of \citet*{angeletos2007dynamic}. Extending our work to a dynamic setting would bring new technical and conceptual challenges.  On the conceptual side, dynamic versions of rationalizability usually incorporate some form of forward or backward induction. On the technical side, dynamic games make it impossible to ignore probability-zero events completely. The existence of extremal equilibria in dynamic supermodular games with purely idiosyncratic uncertainty has been proven by \citet*{BDRWTE2022}. Our approach allows for a much more flexible treatment of beliefs.

An important ingredient of our model is our model of beliefs, and we think it should be useful in wider settings. For example, \citet{morris1995justifying} and \citet{BENPORATH20112608} provide explicit foundations in terms of beliefs for competitive equilibria in economies with incomplete information. In particular, \citet{BENPORATH20112608} study a competitive model with a continuum of agents. However, they use type spaces with a product structure, and this, together with the requirement of measurability, imposes economically overly restrictive type-symmetry conditions on beliefs. Our framework allows for a much more flexible treatment of beliefs.
 
There is an extensive literature on the robustness of equilibria to perturbations of higher-order beliefs; \citet{kajii1997robustness} introduced the general program with ex-ante perturbations and \citet{weinstein2007structure} studied interim perturbations. In our large player setting, perturbations can relate jointly to players and hierarchies of beliefs; topologies on individual hierarchies of beliefs are not sufficient to study perturbations of profiles of hierarchies of beliefs.

\section{Appendix 1: Stochastic orderings}

Let $X$ be a Polish space and $\succeq$ be a partial order on $X$ with a closed graph. We can identify $X$ with a closed subspace of $\Delta(X)$ by the embedding $x\mapsto\delta_x$. We can extend $\succeq$ to all of $\Delta(X)$ by letting $\mu\succeq\mu'$ if $\int g~\mathrm d\mu\geq \int g~\mathrm d\mu'$ for every bounded measurable $\succeq$-nondecreasing function $g:X\to\mathbb{R}$. The following are equivalent:
\begin{enumerate}
\item $\mu\succeq \mu'$.
\item There exists $\lambda\in\Delta(X\times X)$ supported on the graph of $\succeq$ such that the first marginal equals $\mu$ and the second marginal equals $\mu'$.
\item $\mu(U)\geq\mu'(U)$ for every closed set $U\subseteq X$ that contains with every element also all $\succeq$-larger elements.
\item There exists a probability space and random variables $m$ and $m'$ defined on it with values in $X$ such that $m$ has distribution $\mu$, $m'$ has distribution $\mu'$ and $m\succeq m'$ holds almost surely.
\end{enumerate}
For a proof, see \citet*[Theorem 1]{MR0494447}. The $\succeq$ relation is indeed a partial order on all of $\Delta(X)$; see \citet[Theorem 2]{MR512419} and \citet*[Proposition 3]{MR0494447}.

If $\succeq$ is an order on $X$, there is a natural induced partial order on $C\times X$ such that, abusing notation again, $(c,x)\succeq (c',x')$ holds exactly when $c=c'$ and $x\succeq x'$. We will use this order to define stochastic orders on $\Delta_\nu(C\times X)$.

\section{Appendix 2: Proofs}

\begin{proof}[Proof of Lemma \ref{Duhc}]We first show that the set
\[\begin{split}
D=\Big\{\kappa\in\,\Delta_\nu\big(C\times \Delta(T)\times A\big)~\Bigl\lvert~ &\kappa\textnormal{ is supported on the graph of } \mc{S}\Big\}
\end{split}\]
is closed. Now $\Gamma(\mc{S})$, the graph of $\mc{S}$, is closed by the closed graph theorem. Let $d$ be a bounded metric that metrizes $C\times \Delta(T)\times A$. Then
\[\Gamma(\mc{S})=\Big\{(c,\beta,a)\mid d\big((c,\beta,a),\Gamma(\mc{S})\big)=0\Big\}.\]
It follows that \[\begin{split}
D=\Big\{\kappa\in\,\Delta_\nu\big(C\times \Delta(T)\times A\big)~\Bigl\lvert~ & \int d\big(\cdot,\Gamma(\mc{S})\big)~\mathrm d\kappa=0\Big\}
\end{split}\]
and this set is closed as the preimage of $\{0\}$ under a continuous function. Now the graph of $\mathcal{D}_\mc{S}$ is
\[\big\{(t,\kappa)\in T\times D\mid \kappa\textnormal{ has }C\times\Delta(T)\textnormal{-marginal }\tau(t)\big\}.\]
This graph is closed since both the ``marginal-mapping'' and $\tau$ are continuous functions. It follows that $\mc{D}_\mc{S}$ is upper hemicontinuous with compact values since $T\times\Delta_\nu\big(C\times \Delta(T)\times A\big)$ is compact. That all values are nonempty is obvious.
\end{proof}

\begin{proof}[Proof of Lemma \ref{Sdef}]First note that the set
\[\begin{split}
\bigg\{(c,a,\beta,\mu)\in\, & C\times A\times\Delta(T)\times\Delta\left(T\times\Delta_\nu\big(C\times \Delta(T)\times A\big)\right)\Bigl\lvert\\  &a\in\argmax_{\mc{A}(c)} \int v\big(c,a,\sigma(t),\phi(\kappa)\big)~\mathrm d\mu(t,\kappa),\\
&\mu\textnormal{ has }T\textnormal{-marginal }\beta.
\bigg\}
\end{split}\]
is closed because expected utility is jointly continuous (see for example \citet*{MR2221405}) and the marginal function is continuous. This set is the graph of a correspondence from
$C\times\Delta(T)$ to the compact set $A\times \Delta\big(T\times\Delta_\nu(C\times \Delta(T)\times A)\big)$ and this correspondence is upper hemicontinuous by the closed graph theorem. Composing with the projection of $A\times\Delta\big(T\times\Delta_\nu(C\times \Delta(T)\times A)\big)$ onto $A$ gives us the correspondence $\mc{S}_\mc{D}$ which is therefore upper hemicontinuos by \citet[Theorem 17.23]{MR2378491}. Since forward images of compact sets under continuous functions are compact, $\mc{S}_\mc{D}$ is compact valued.
\end{proof}

\begin{proof}[Proof of Theorem \ref{selfrat}]Let $a\in\mathcal{S}(c,\beta)$. For each positive $m$, there exists $\mu_m\in\Delta\big(T\times\Delta_\nu(C\times \Delta(T)\times A)\big)$ with $T$-marginal $\beta$, supported on the graph of $\mc{D}_{m-1}$ such that $a\in\argmax_{\mathcal{A}(c)} \int v\big(c,a,\sigma(t),\phi(\kappa)\big)~\mathrm d\mu(t,\kappa).$ Since $\Delta\big(T\times\Delta_\nu(C\times \Delta(T)\times A)\big)$ is compact, there exists a subsequence converging to some $\mu$. We claim that $\mu$ is supported on the graph of $\mc{D}_\mc{S}$ and that \[a\in\argmax_{\mathcal{A}(c)} \int v\big(c,a,\sigma(t),\phi(\kappa)\big)~\mathrm{d}\mu(t,\kappa).\] Indeed, the latter part follows again from the joint continuity of expected utility.

To see that $\mu$ is supported on the graph of $\mc{D}_\mc{S}$, we note that we only have to check that the $\Delta_\nu(C\times \Delta(T)\times A)\big)$-marginal is supported on the graph of $\mc{S}$. Let $O_m$ be the complement of the graph of $\mathcal{S}_m$. Let $\mu_m^*$ and $\mu^*$ be the $\Delta_\nu(C\times\Delta(T)\times A)$-marginal of $\mu_m$ and $\mu$ respectively. Clearly, $\mu_n^*(O_m)=0$ for $n\geq m$. Therefore, $0=\liminf_n \mu_n^*(O_m)\geq \mu^*(O_m)=0$ by the Portmanteua theorem. Now $\bigcup_m O_m$ is exactly the complement of the graph of $\mathcal{S}$ and by countable subadditivity, $\mu^*(\bigcup_m O_m)=0$.

Now, if $\mathcal{R}$ rationalizes itself, it follows from a simple induction argument that $\mathcal{R}(c,\beta)\subseteq\mathcal{S}_m(c,\beta)$ for all $(c,\beta)$, so $\mathcal{S}$ is the largest correspondence that rationalizes itself.
\end{proof}

For the proof of Theorem \ref{equilibrium}, we need some preliminary work. We define $V:C\times A\times\Delta\big(S\times \Delta_\nu(C\times A)\big) \to\mathbb{R}$ by
\[V(c,a,\kappa)=\int v(c,a,s,\mu)~\mathrm d\kappa(s,\mu).\]

\begin{lemma}\label{incdiff}The function $V$ is supermodular in $A$ and has increasing differences in $A\times\Delta\big(S\times \Delta_\nu(C\times A)\big)$.
\end{lemma}
\begin{proof}Verifying that $V$ is supermodular in $A$ is straightforward. For the rest, assume that $a\succeq a'$ and $\kappa\succeq\lambda$. By the characterization of stochastic dominance in Appendix 1, there exists a probability space and random variables $k$ and $l$ defined on it with values in $S\times\Delta_\nu(C\times A)$ such that $k\succeq l$ almost surely and such that the distributions of $k$ and $l$ are $\kappa$ and $\lambda$, respectively. By assumption, we have
\[v(c,a,k)-v(c,a',k)\geq v(c,a,l)-v(c,a',l)\]
almost surely. By the change of variables formula for push-forward measures, the integral of $v(c,a,k)-v(c,a',k)$ is $V(c,a,\kappa)-V(c,a,\kappa)$ and the integral of $v(c,a,l)-v(c,a',l)$ is $V(c,a,\lambda)-V(c,a',\lambda)$, so the result follows from integrating the inequality $v(c,a,k)-v(c,a',k)\geq v(c,a,l)-v(c,a',l)$.
\end{proof}



\begin{lemma}\label{extrselect}For all $m$, $\bigvee\mc{S}_m$ and $\bigwedge\mc{S}_m$ are measurable selections of $\mc{S}_m$
\end{lemma}
\begin{proof}We do the case of $\bigvee \mc{S}_m$ by induction. Note first that for every measurable correspondence $W$ whose values are nonempty compact sublattices of $A$, the function $\bigvee W$ is a measurable selection of $W$, see \citet[Lemma 19]{VanZandt2010JET}.

So $\bigvee\mc{S}_0$ is a measurable selection of $\mc{S}_0$ since the values of $\mc{A}$ are compact sublattices of $A$. Now assume that $\bigvee\mc{S}_m$ is a measurable selection of $\mc{S}_m$. Let $(c,\beta)\in C\times\Delta(T)$ and let $\mu\in\Delta\big(T\times\Delta_\nu(C\times\Delta(T)\times A)\big)$ have $T$-marginal $\beta$ and be supported on the graph of $\mc{D}_m$. We show that $\mu\circ(\sigma,\phi)^{-1}\in\Delta\big(S\times\Delta_\nu(C\times A)\big)$ satisfies \[\beta\circ\bigg(\sigma,\phi\circ{\bigvee\mc{S}_m}^*\circ\tau\bigg)^{-1}\succeq\mu\circ (\sigma,\phi)^{-1}.\]
So let $g:\Delta\big(S\times\Delta_\nu(C\times A)\big)\to\mathbb{R}$ be a bounded measurable function increasing in $\Delta_\nu(C\times A)$. Let $r_\mu:T\to\Delta\big(\Delta_\nu(C\times\Delta(T)\times A)\big)$ be a regular conditional probability for $\mu$. We have
\[\int g(s,\kappa)~\mathrm d\beta\circ\bigg(\sigma,\phi\circ{\bigvee\mc{S}_m}^*\circ\tau\bigg)^{-1}=\int g\bigg(\sigma(t),\phi\Big({\bigvee{\mc{S}_m}^*}\big(\tau(t)\big)\Big)\bigg)~\mathrm d\beta(t)\]
and
\[\int g(s,\kappa)~\mathrm d\mu\circ(\sigma,\phi)^{-1}=\int g\big(\sigma(t),\phi(\kappa)\big)~\mathrm d\mu(t,\kappa)\]
\[=\int\int g\big(\sigma(t),\phi(\kappa)\big)~\mathrm d r_\mu~\mathrm d\beta.\]
Since $\mu$ is supported on the graph of $\mc{D}_m$, $r_\mu(t)$ must be supported on the $t$-section of this graph for $\beta$-almost all $t$. It follows that $r_\mu(t)$ assigns probability one to elements of $\Delta_\nu(C\times\Delta(T)\times A)$ with $C\times\Delta(T)$-marginal $\tau(t)$ that are supported on $\mc{S}_m(\tau(t))$ for $\beta$-almost all $t$. But $\bigvee\mc{S}_m(\tau(t))$ is at least as large, so
\[{\bigvee\mc{S}_m}^*\circ\tau(t)\succeq\kappa\]
for $r_\mu(t)$ almost all $\kappa$. It follows then by integration and the appropriate monotonicity properties of $\phi$ and $g$ that
\[\int g(s,\kappa)~\mathrm d\beta\circ\bigg(\sigma,\phi\circ{\bigvee\mc{S}_m}^*\circ\tau\bigg)^{-1}\geq\int g(s,\kappa)~\mathrm d\mu\circ(\sigma,\phi)^{-1}.\]
By Lemma \ref{incdiff} and \citet[Corollary 4.1]{Topkis78},
\[\argmax_{\mc{A}(c)}V\Bigg(c,\cdot,\beta\circ\bigg(\sigma,\phi\circ{\bigvee\mc{S}_m}^*\circ\tau\bigg)^{-1}\Bigg)\]
is a complete sublattice and clearly compact. By \citet[Lemma 19]{VanZandt2010JET}, its supremum, which is a measurable selection and therefore also a measurable selection of $\mc{S}_{m+1}$, is larger than every element of $\mc{S}_{m+1}$. But this means it must equal $\bigvee\mc{S}_{m+1}$ and $\bigvee\mc{S}_{m+1}$ is a measurable selection of $\mc{S}_{m+1}$.
\end{proof}

\begin{proof}[Proof of Theorem \ref{equilibrium}]We do the first case. We first show that $\bigvee\mc{S}$ is a measurable selection of $\mc{S}$. Indeed, since $\langle\mc{S}_m\rangle$ is a decreasing sequence of correspondences with nonempty and compact values, and since $\bigvee\mc{S}_m$ is a selection of $\mc{S}_m$ for each $m$ by Lemma \ref{extrselect}, $\lim_{m\to\infty}\bigvee\mc{S}_m(c,\beta)$ exists for all $(c,\beta)$ and is an element of $\mc{S}(c,\beta)$. Since $\bigvee\mc{S}_m(c,\beta)$ is an upper bound of $\mc{S}(c,\beta)$ for all $m$ and the order $\succeq$ on $A$ has a closed graph, $\bigvee\mc{S}(c,\beta)=\lim_{m\to\infty}\bigvee\mc{S}_m(c,\beta)$. It now follows that $\bigvee\mc{S}$ is measurable as the pointwise limit of a sequence of measurable functions. So $\bigvee\mc{S}$ is a measurable selection of $\mc{S}$.

It remains to show that the strategy profile $\bigvee\mc{S}$ is a best response to itself. Since $\mc{S}$ rationalizes itself by Theorem \ref{selfrat}, there exists for each $(c,\beta)$ some $\mu\in\Delta\big(T\times\Delta_\nu(C\times \Delta(T)\times A)\big)$ with $T$-marginal $\beta$ and supported on the graph of $\mc{D}_\mc{S}$ such that
\[\bigvee\mc{S}(c,\beta)\in\argmax_{\mathcal{A}(c)} \int v\big(c,\cdot,\sigma(t),\phi(\kappa)\big)~\mathrm d\mu(t,\kappa)
.\]
We show that $\mu\circ(\sigma,\phi)^{-1}\in\Delta\big(S\times\Delta_\nu(C\times A)\big)$ satisfies \[\beta\circ\bigg(\sigma,\phi\circ{\bigvee\mc{S}}^*\circ\tau\bigg)^{-1}\succeq\mu\circ (\sigma,\phi)^{-1}.\]
%

So let $g:\Delta\big(S\times\Delta_\nu(C\times A)\big)\to\mathbb{R}$ be a bounded measurable function increasing in $\Delta_\nu(C\times A)$. Let $r_\mu:T\to\Delta\big(\Delta_\nu(C\times\Delta(T)\times A)\big)$ be a regular conditional probability for $\mu$. We have
\[\int g(s,\kappa)~\mathrm d\beta\circ\bigg(\sigma,\phi\circ{\bigvee\mc{S}}^*\circ\tau\bigg)^{-1}=\int g\bigg(\sigma(t),\phi\Big({\bigvee\mc{S}^*}\big(\tau(t)\big)\Big)\bigg)~\mathrm d\beta(t)\]
and
\[\int g(s,\kappa)~\mathrm d\mu\circ(\sigma,\phi)^{-1}=\int g\big(\sigma(t),\phi(\kappa)\big)~\mathrm d\mu(t,\kappa)\]
\[=\int\int g\big(\sigma(t),\phi(\kappa)\big)~\mathrm d r_\mu~\mathrm d\beta.\]
Since $\mu$ is supported on the graph of $\mc{D}_\mc{S}$, $r_\mu(t)$ must be supported on the $t$-section of this graph for $\beta$-almost all $t$. It follows that $r_\mu(t)$ assigns probability one to elements of $\Delta_\nu(C\times\Delta(T)\times A)$ with $C\times\Delta(T)$-marginal $\tau(t)$ that are supported on $\mc{S}(\tau(t))$ for $\beta$-almost all $t$. But $\bigvee\mc{S}(\tau(t))$ is at least as large, so
\[{\bigvee\mc{S}}^*\circ\tau(t)\succeq\kappa\]
for $r_\mu(t)$ almost all $\kappa$. It follows then by integration and the appropriate monotonicity properties of $\phi$ and $g$ that
\[\int g(s,\kappa)~\mathrm d\beta\circ\bigg(\sigma,\phi\circ{\bigvee\mc{S}}^*\circ\tau\bigg)^{-1}\geq\int g(s,\kappa)~\mathrm d\mu\circ(\sigma,\phi)^{-1}.\]
Now, for all $a\in\mc{A}(c)$,
\[\int v\big(c,a,\sigma(t),\phi(\kappa)\big)~\mathrm d\mu(t,\kappa)=\int v(c,a,\cdot)~\mathrm d\mu\circ(\sigma,\phi)^{-1}=V(c,a,\mu\circ(\sigma,\phi)^{-1})\]
by the change of variables formula for pushforward-measures. So \[\bigvee\mc{S}(c,\beta)\in\argmax_{\mathcal{A}(c)}V(c,a,\mu\circ(\sigma,\phi)^{-1}).\]
Let \[a\in\argmax_{\mathcal{A}(c)}V\Bigg(c,a,\beta\circ\bigg(\sigma,\phi\circ{\bigvee\mc{S}}^*\circ\tau\bigg)^{-1}\bigg)\Bigg).\]
So by Lemma \ref{incdiff}, \citet[Theorem 6.1]{Topkis78}, and what we have shown above,
 \[a\vee\bigvee\mc{S}(c,\beta)\in\argmax_{\mathcal{A}(c)}V\Bigg(c,a,\beta\circ\bigg(\sigma,\phi\circ{\bigvee\mc{S}}^*\circ\tau\bigg)^{-1}\bigg)\Bigg).\]
Since $\mc{S}$ rationalizes itself $a\vee\bigvee\mc{S}(c,\beta)\in\mc{S}(c,\beta)$. Therefore, $\bigvee\mc{S}(c,\beta)\succeq a\vee\bigvee\mc{S}(c,\beta)$. Since we, trivially, also have $a\vee\bigvee\mc{S}(c,\beta)\succeq\bigvee\mc{S}(c,\beta)$, we have $a\vee\bigvee\mc{S}=\bigvee\mc{S}$ and therefore
 \[\bigvee\mc{S}(c,\beta)\in\argmax_{\mathcal{A}(c)}V\Bigg(c,a,\beta\circ\bigg(\sigma,\phi\circ{\bigvee\mc{S}}^*\circ\tau\bigg)^{-1}\bigg)\Bigg),\]
which shows that $\bigvee\mc{S}$ is a Bayes-Nash equilibrium.
\end{proof}

\section{Appendix 3: The canonical type space}\label{A:Canonical}

We define coherent types first for abstract hierarchies while ignoring what these are about. We let $Z_0,Z_1,\ldots$ be a sequence of nonempty Polish spaces and write $X_n=\prod_{k=0}^n Z_k$ and $Z=\prod_{k=0}^\infty Z_k$. A nonempty product space is compact if and only if every factor space is, so $Z$ is compact if and only if each $Z_k$ is. An \emph{abstract type} is an element of $T_0=\prod_{k=1}^{\infty}\Delta(X_{k-1})$. The abstract type $\langle\delta_k\rangle_{k=1}^\infty$ is \emph{coherent} if $\delta_k=\marg_{X_{k-1}}\delta_{k+1}$ for all $k$. We write $T_1$ for the subspace of $T_1$ consisting of all coherent abstract types. The set of Borel probability measures on a Polish space is compact if and only if the underlying Polish space is. Consequently, $T_0$ is compact if and only if every $Z_k$ is.

\begin{lemma}\label{abstcoh}The function $\chi:\Delta(Z)\to T_1$ given by 
\[\chi(\delta)=\langle\marg_{X_{k-1}}\delta\rangle_{k=1}^\infty\]
is a homeomorphism.
\end{lemma}
\begin{proof}Clearly, every element of the range of $\chi$ is coherent by construction. The continuity of $\chi$ follows from the continuity of the marginal mappings. Since a probability measure on a product space is determined by the finite dimensional marginals (cylinder sets form a $\pi$-system), $\chi$ is injective. The Kolmogorov extension theorem, \citet[Theorem 15.26]{MR2378491}, implies that $\chi$ is surjective. The continuity of $\chi^{-1}$ follows from \citet[Corollary 2.4.8]{MR3837546}.
\end{proof}
 For generic $\langle\delta_k\rangle_{k=1}^{\infty}\in T_1'$, we write $\delta$ for the unique probability measure mapped to it by $\chi$.\bigskip

We let $S$ and $C$ be nonempty Polish spaces, $\nu$ a Borel probability measure on $C$ and write $\Delta_\nu(C\times Y)$ for the set of all probability measures on $C\times Y$ with $C$-marginal $\nu$. We now let $Z_0=X_0=S$ and define recursively 
\[Z_{n+1}=\Delta_\nu\big(C\times \Delta(X_0)\times\cdots\times\Delta(X_{n-1})\big)=\Delta_\nu\bigg(C\times\prod_{k=0}^{n-1}\Delta(X_0)\bigg).\] We call abstract types in this setting simply \emph{types}. A type represents an individual agent's hierarchy of beliefs on the population distribution of finite hierarchies of beliefs.\bigskip

Next, we let $Z'_k=Z_{k+1}$ for all $k$ and  $Z'=\prod_{k=0}^\infty Z_k'$. $T_0'$ and $T_1'$ are defined correspondingly. By Lemma \ref{abstcoh}, there is a canonical homeomorphism between $\Delta(Z')$ and $T_1'$. We can exploit the special structure of $T_1'$ to prove the following:
\begin{lemma}\label{soccoh}The function \[\chi':\Delta_\nu(C\times T_0)\to T_1'\] given by 
\[\chi'(\lambda)=\Big\langle\marg_{C\times\prod_{j=0}^{k-1}\Delta(X_j)}\lambda\Big\rangle_{k=1}^\infty\]
is a homeomorphism.
\end{lemma}
\begin{proof}The proof follows the proof of Lemma \ref{abstcoh} almost verbatim, the only difference is in how one applies the Kolmogorov extension theorem. 
\end{proof}
Lemma \ref{soccoh} allows us to identify $T_1'$ with a population distribution over types. These types need not be coherent. We want to incorporate the restrictions that the types are coherent, that they are certain that the population distribution of types consists only of coherent types, and so on. Dealing with this issue is the next step. For generic $\langle\lambda_k\rangle_{k=1}^{\infty}\in T_1'$, we write $\lambda$ for the unique probability measure mapped to it by $\chi'$. For a measurable set $E\subseteq T_0$, let
\[B(E)=\big\{\langle\lambda_k\rangle_{k=1}^{\infty}\in T_1'\mid \lambda(C\times E)=1\big\}.\]
If $E$ is closed, so is $B(E)$. We can now recursively define $T_n$ for positive $n$ by
\[T_n=\Big\{\langle\delta_k\rangle_{k=1}^\infty\in T_1\mid \delta\big(S\times B(T_{n-1})\big)=1\Big\}\]
and let $T_\infty=\bigcap_n T_n$.

\begin{lemma}\label{FP}The set $T_\infty$ is nonempty, closed, and 
\[T_\infty=\Big\{\langle\delta_k\rangle_{k=1}^\infty\in T_1\mid \delta\big(S\times B(T_\infty)\big)=1\Big\}.\]
\end{lemma}
\begin{proof}We first show that $\langle T_n\rangle$ is a decreasing sequence. Clearly, $T_0\supseteq T_1$ are closed sets. Since $B$ is an increasing operator that maps closed sets to closed sets, $T_{n-1}\supseteq T_n$ implies
\[T_{n+1}=\Big\{\langle\delta_k\rangle_{k=1}^\infty\in T_1\mid \delta\big(S\times B(T_n)\big)=1\Big\}\subseteq\]
\[\Big\{\langle\delta_k\rangle_{k=1}^\infty\in T_1\mid \delta\big(S\times B(T_{n-1})\big)=1\Big\}=T_n.\]
Consequently, we have for each $\langle\delta_k\rangle_{k=1}^\infty\in T_\infty$ that $\delta\big(S\times B(T_n)\big)=1$ holds for all $n$ and, therefore, also
 $\delta\big(S\times B(T_\infty)\big)=1$. Conversely, $\delta\big(S\times B(T_\infty)\big)=1$ implies trivially that $\delta\big(S\times B(T_n)\big)=1$ for all $n$. This gives us the equality
\[T_\infty=\Big\{\langle\delta_k\rangle_{k=1}^\infty\in T_1\mid \delta\big(S\times B(T_\infty)\big)=1\Big\}.\]
Also, $T_\infty$ is closed as an intersection of closed sets. To see that $T_\infty$ is nonempty, take some $s\in S$ and replace $S$ by the singleton $\{s\}$. For the resulting model, there is a single hierarchy in $T_0$, which is trivially in every $T_n$, even for the original $S$. 
\end{proof}

\begin{proposition}\label{canhomeom}$T_\infty$ is the homeomorphic image of $\Delta\big(S\times\Delta_\nu(C\times T_\infty)\big)$ under the mapping given by
\[\alpha \mapsto \chi\big(\alpha\circ(\id_S,\chi')^{-1}\big).\]
\end{proposition}
\begin{proof}By Lemma \ref{FP}, $T_\infty$ is the homeomorphic image of $\Delta\big(S\times B(T_\infty)\big)$ under $\chi$ and, by definition,
\[B(T_\infty)=\big\{\langle\lambda_k\rangle_{k=1}^{\infty}\in T_1'\mid \lambda(C\times T_\infty)=1\big\},\]
so $B(T_\infty)$ is the homeomorphic image of $\Delta_\nu(C\times T_\infty)$ under $\chi'$. 
\end{proof}

We use these results to construct a particular type space. Let $T^*=S\times\Delta_\nu(C\times T_\infty)$ and let $\sigma^*:T^*\to S$ be the projection onto $S$. Finally, let $\tau^*:T^*\to\Delta_\nu(C\times \Delta(T^*))$ be the composition of the inverse of the homeomorphism defined in the statement of Proposition \ref{canhomeom} with the  projection onto the second coordinate. We call $(T^*,\sigma^*,\tau^*)$ the \emph{canonical type space}. Note that $\sigma^*$ and $\tau^*$ are continuous, and $T^*$ is compact whenever $S$ is.

\small
\bibliography{ICR}

\end{document}